\documentclass{llncs}
\usepackage{makeidx,enumerate,graphicx,amsfonts,amssymb}
\usepackage{graphicx,epsfig}
\usepackage[lined]{algorithm2e}
\usepackage{fontenc}
\usepackage[utf8]{inputenc}
\usepackage{mathtools}

\newtheorem{observation}[theorem]{Observation}

\begin{document}

\title{Completely Positive formulation of the Graph Isomorphism Problem}
\author{Shashank K Mehta \and Pawan Aurora }
\institute{Indian Institute of Technology, Kanpur - 208016, India\\
\email{skmehta@cse.iitk.ac.in}, paurora@cse.iitk.ac.in}
\tocauthor{Shashank K Mehta (Indian Institute of Technology, Kanpur) }

\maketitle

\begin{abstract}
Given two graphs $G_1$ and $G_2$ on $n$ vertices each, we define a graph $G$ on vertex set $V_1\times V_2$
and the edge set as the union of edges of $G_1\times \bar{G_2}$, $\bar{G_1}\times G_2$, 
$\{(v,u'),(v,u''))(|u',u''\in V_2\}$ for each $v\in V_1$, and $\{((u',v),(u'',v))|u',u''\in V_1\}$ 
for each $v\in V_2$.
We consider the completely-positive Lov\'asz $\vartheta$ function, i.e., $cp\vartheta$ function for $G$. 
We show that the function evaluates to $n$ whenever $G_1$ and $G_2$ are isomorphic and to less than
$n-1/(4n^4)$ when non-isomorphic. Hence this function provides a test
for graph isomorphism. We also provide some geometric insight into the feasible 
region of the completely positive program.
\end{abstract}

\section{Introduction}
Let $G_1=(V,E_1),G_2=(V,E_2)$ be two simple undirected graphs, where $V$ is the set of 
vertices of cardinality $n$ and $E_1,E_2$ are the respective sets of edges. $G_1$ and $G_2$ 
are called isomorphic if there exists a bijection $\sigma:V\rightarrow V$ such that 
$(\sigma(x),\sigma(y))\in E_2$ if and only if $(x,y)\in E_1$. 
 The graph isomorphism problem (GI) is the problem of determining if $G_1$ and $G_2$ are 
isomorphic. This problem although clearly in class \textbf{NP}, has not been known to be 
either in \textbf{P} or \textbf{NP-Complete} \cite{GJ}, except for certain graphs where it is known to have polynomial complexity \cite{BGM,Bod,FM,HW,Luks,Miller}.

There has been evidence suggesting that GI is not likely to be \textbf{NP-Complete}. One of them being that its counting version is reducible to its decision version \cite{Mathon}. Moreover, if the problem were \textbf{NP-Complete}, then the polynomial time hierarchy would collapse to its second level \cite{Babai,BHZ,Schon}. A lot of research has therefore gone into determining the largest complexity class for which it can be shown that GI is hard \cite{JKMT,Toran}. The largest complexity class known to be reducible to GI is \textbf{DET} \cite{Toran}. The complexity aspects of GI are treated in much detail in \cite{AT,KST}.

Apart from the obvious theoretical importance of determining its computational complexity, the graph isomorphism problem finds such diverse applications as chemical identification \cite{Sussen}, scene analysis \cite{ABBBP} and construction and enumeration of combinatorial configurations \cite{CM}.

Several approaches to solve GI in polynomial time have been adopted. Among them is an approach to
incrementally build an isomorphism between the graphs, \cite{RND}. Another approach has been 
to find a canonical labeling of the vertices of the two graphs, \cite{BES,BK,McKay}. A comprehensive 
list of all the approaches is difficult to present here. There are some survey papers on the work
published on this problem, such as \cite{Fortin}.

It was conjectured in \cite{Ramana} that the graph isomorphism problem can be reduced to 
a semidefinite feasibility problem. We make an attempt in this direction.

Given two graphs $G_1$ and $G_2$ on $n$ vertices each, we consider the Lov\'asz $\vartheta$ 
function, for an $n^2$ vertex graph based on the two input graphs, with positive semidefinite
condition replaced by completely positivity condition.
We show that if the graphs are isomorphic then the function evaluates to $n$ and if 
non-isomorphic then it evaluates to a value less than $n-1/(4n^4)$. Hence this provides a test for GI.

\section{Preliminaries}

\subsection{Positive Semidefinite Matrices}

An $m\times m$ symmetric matrix $M$ is said to be {\em positive semidefinite} if it can be
expressed as $Q\cdot Q^T$ for some $m\times k$ matrix $Q$. If the row vectors of $Q$ 
are $v_1,\dots,v_m$, then we will call this set a vector-realization of $M$ in $k$-dimensional
space. We will denote the corresponding matrix $M$ by ${\bf M}(v_1,\dots,v_m)$.
It is easy to see that there is always a vector realization in $k= rank(M)$ dimensional
space.

If all entries of a positive semidefinite matrix $M$ are non-negative, then it is called
a {\em doubly-non-negative} (DN) matrix. Further, if $M$ has a vector realization in which 
each component of each vector is non-negative (i.e., $M$ has a decomposition $Q\cdot Q^T$
where each entry of $Q$ is non-negative), then $M$ is called a {\em completely positive}
(CP) matrix. We will call it a non-negative vector realization of the CP matrix $M$.
Every principal submatrix of a DN (resp. CP) matrix is DN (resp. CP).
It is easy to see that every CP matrix is a DN matrix.

It is not necessary that every decomposition $Q\cdot Q^T$ of a CP matrix
has all non-negative entries in $Q$. The smallest $k$, for which such an $m\times k$ matrix 
exists, is called the {\em cp-rank} of $M$. 

\begin{theorem}[\cite{Hannah}] \label{thm1}
For any CP matrix $M$ of rank $r$, $cp$-$rank(M) \leq r(r+1)/2$.
\end{theorem}

A geometrical view of a CP matrix is that if $v_1,\dots,v_m$ is a non-negative vector realization
of it, then these vectors belong to the {\em closed positive orthant} (`closed' in the sense of a 
polyhedron) of some orthogonal basis of the space.

\subsection{United Vectors}

Let $w$ be any fixed unit vector. Then for every unit vector $v$, we call $u=(w+v)/2$
a {\em united vector with respect to $w$}. We will drop the reference to $w$ when it is
unambiguous.

\begin{observation}\label{obs1} With respect to a fixed unit vector $w$,\\
(i) a vector $u$ is united if and only if $u\cdot w = u^2$,\\
(ii) if $u_1$ and $u_2$ are mutually orthogonal united vectors,
then $u_1+u_2$ is also a united vector.\\
(iii) let $u_1,\dots,u_k$ be a set of pairwise orthogonal united vectors.
This set is maximal (i.e., no new united vector
can be added to it while preserving pairwise orthogonality) if and only if $w$ belongs
to the subspace spanned by these vectors if and only if $\sum_i u_i = w$.\\
(iv) for any collection of pairwise orthogonal united vectors $u_1,\dots, u_j$,
$w\cdot \sum_i u_i = \sum_i u_i^2 \leq 1$. Further, $\sum_i u_i^2 = 1$ if and only if 
the set is maximal, i.e., $\sum_i u_i = w$.
\end{observation}

\begin{lemma}\label{lem1} Let $u_1,\dots, u_k$ be pairwise orthogonal united vectors. Then
matrices ${\bf M}(u_1,\dots,u_k)$ and ${\bf M}(u_1,\dots,u_k,w)$ are CP.
\end{lemma}

\begin{proof} Suppose $u_i\cdot w=a_i$ for $i=1,\dots,k$. Then there exists a coordinate
system in $(k+1)$-dimensional space in which $u_1=(\sqrt{a_1},0,\dots,0)$, $u_2=(0,\sqrt{a_2},0,\dots,0)$
so on, and $w=(\sqrt{a_1},\sqrt{a_2},\dots,\sqrt{a_k},\sqrt{1-\sum_ia_i})$. 
Since all entries are non-negative reals, the claim is established. $\Box$
\end{proof}

\section{Lov\'asz Theta Function \cite{LTF}}
Given two graphs, each on $n$ vertices, $G_1=([n],E_1)$ and $G_2=([n],E_2)$. consider 
the semidefinite program SDP-LT given below, which computes a $\vartheta$-function. The
variable matrix $Y$ is of size $(n^2+1)\times (n^2+1)$ with index set $\{ij | i,j\in [n]\}
\cup \{\omega\}$.

\begin{alignat}{2}
    \text{SDP-LT:}\quad\text{maximize } \quad & \sum_{i,j\in [n]}Y_{ij,ij}\  \notag\\
    \text{subject to} \quad & Y\succeq 0 \ \\
                       & Y_{ij,kl}\geq 0\ &,\ & 1\leq i,j,k,l\leq n\\
                       & Y_{\omega,\omega}=1 \ \\
                       & Y_{ij,\omega}=Y_{ij,ij}\ &,\ & 1\leq i,j\leq n\\
                       & Y_{ij,ik}=0\ &,\ & 1\leq i,j,k\leq n,\ j\neq k\\
                       & Y_{ji,ki}=0\ &,\ & 1\leq i,j,k\leq n,\ j\neq k\\
                       & Y_{ij,kl}=0\ &,\ & (i,k)\in E_1 ,\ (j,l)\notin E_2\\
                       & Y_{ij,kl}=0\ &,\ & (i,k)\notin E_1 ,\ (j,l)\in E_2
  \end{alignat}

Let $Y$ be a solution of SDP-LT and let $\{u_{ij}|i,j\in[n]\} \cup \{w\}$ be a vector realization
of $Y$. Then from conditions (3) and (4) every $u_{ij}$ is a united vector with respect to the
unit vector $w$. Conditions (1) and (2) ensure that $Y$ is DN.

Every solution matrix $Y$ of SDP-LT is $(n^2+1)\times (n^2+1)$ in size
in which the last row and the last column are same as the diagonal. Hence from here onwards
we will drop the last row and the last column and assume that $Y$ is an $n^2\times n^2$ matrix.

Let $\{u_{ij}|i,j\in [n]\}\cup \{w\}$ be a vector realization of any solution $Y$ of SDP-LT.
Consider the matrix 
\[ W=\left( \begin{array}{cccc}
u_{11} & u_{12} & \dots & u_{1n}\\
u_{21} & u_{22} & \dots & u_{2n}\\
\vdots  & \vdots  & \ddots & \vdots\\
u_{n1} & u_{n2} & \dots & u_{nn} \end{array} \right)\] 
Each row and each column of this matrix is a set of pairwise orthogonal united vectors 
with respect to $w$. Hence from Observation \ref{obs1}(iv) the value of the objective function of SDP-LT, 
$\sum_{i,j\in[n]}u_{ij}^2$, is at most $n$.
\vspace{1mm}

\noindent
{\bf Remark:} The graph for which SDP-LT is a Lov\'asz $\vartheta$ function, has a clique cover of 
size $n$ (the edges of condition (5)). Hence this also establishes that the function value 
is bounded above by $n$.

\section{The Case of Isomorphic Graphs}

Let us continue to assume that $\{u_{ij}|i,j\in [n]\}\cup \{w\}$ is a vector realization 
of an arbitrary solution $Y$ of SDP-LT. Suppose there exists any set of $n$
vectors $\{u_{i_1j_1},\dots, u_{i_nj_n}\}$ in which each pair has positive inner product.
We will call such a set a {\em complete consistent set}. Observe that
both $i_1,\dots,i_n$ and $j_1,\dots,j_n$ are permutations of $1,2,\dots,n$. Hence we can
rearrange them as $\{u_{1\sigma(1)},u_{2\sigma(2)},\dots,\\
 u_{n\sigma(n)}\}$. It is easy to see that
in this case $\sigma$ is an isomorphism between $G_1$ and $G_2$.

Consider the case when $\sigma$ is an isomorphism between $G_1$ and $G_2$. Consider a 
special solution of SDP-LT for this case: Let $w_0$ be some constant unit vector. Define
$w=w_0$ and $u_{i\sigma(i)} = w$ for all $i$ and $u_{ij}=0$
whenever $j\neq \sigma(i)$. In this case the matrix $Y$ is $P^{[2]}_{\sigma}$ defined below.

\begin{definition} For any permutation $\sigma \in S_n$ (the symmetric group), 
the $n^2\times n^2$ matrix $P^{[2]}_{\sigma}$
is defined by $[P^{[2]}_{\sigma}]_{ij,kl} = [P_{\sigma}]_{ij}\cdot [P_{\sigma}]_{kl}$, 
where $P_{\sigma}$ denotes the permutation matrix of $\sigma$. 
The convex hull of $\{P^{[2]}_{\sigma}| \sigma \textrm{ is a } G_1,G_2 
\textrm{ isomorphism}\}$ will be denoted by ${\cal P}_{G_1,G_2}$.
\end{definition}

\begin{observation}\label{obs2} By construction $P^{[2]}_{\sigma}$ matrices are rank-1 positive
semidefinite matrices. Since all entries of $P_{\sigma}$ are non-negative, $P^{[2]}_{\sigma}$ are CP.
\end{observation}

If $P^{[2]}_{\sigma}$ is a solution  of SDP-LT, then $\sigma$ is an isomorphism because
$(P^{[2]}_{\sigma})_{i\sigma(i),j\sigma(j)}\\
 = 1$ for all $i,j$.
Above discussion leads to the following lemma.

\begin{lemma}\label{lem3} $P^{[2]}_{\sigma}$ is a CP solution of SDP-LT if and only if $\sigma$ is 
an isomorphism between $G_1$ and $G_2$.
\end{lemma}

The value of the objective function for $Y=P^{[2]}_{\sigma}$ is 
$\sum_{ij}[P^{[2]}_{\sigma}]_{i\sigma(i),i\sigma(i)} = n$. From the last statement
of the previous section we have the following result.

\begin{lemma}\label{lem4} The value of the objective function of SDP-LT is less than or equal to $n$. 
It reaches its maximum value $n$ when $G_1$ and $G_2$ are isomorphic.
\end{lemma}

\begin{lemma}\label{lem5} Let $Y=\sum_{\sigma \in I} a_{\sigma} P^{[2]}_{\sigma}$ is a 
solution of SDP-LT, where $a_{\sigma}>0$ for each $\sigma \in I$ and $\sum_{\sigma\in I}a_{\sigma}=1$. 
Then $P^{[2]}_{\sigma}$ is a solution of SDP-LT for each $\sigma \in I$.
\end{lemma}

\begin{proof} When $P^{[2]}_{\sigma}$ is extended to $(n^2+1)\times (n^2+1)$, then conditions
(3) and (4) of SDP-LT are trivially satisfied. The extended matrix is equal to $Q\cdot Q^T$ where
$(n^2+1) \times 1$ matrix $Q$ has first $n^2$ entries same as those of $P_{\sigma}$ (i.e., 
first $n^2$ entries is the vectorized $P_{\sigma}$) and the last entry is $1$. Hence it satisfies
conditions (1) and (2). The last four conditions of the SDP are satisfied by the extended
$P^{[2]}_{\sigma}$ because every zero condition satisfied by $Y$ is also satisfied by $P^{[2]}_{\sigma}$.
Thus $P^{[2]}_{\sigma}$ is a solution for every $\sigma\in I$.
$\Box$
\end{proof}

From now on we consider SDP-LT with conditions (1) and (2) replaced with the condition that 
$Y\in\mathcal{C}^{*}$ where $\mathcal{C}^{*}$ is the cone of Completely Positive matrices. 
Let us call the modified program CP-LT and denote the function by $cp\vartheta$.

Now we present the main result of this section.

\begin{lemma}\label{lem6} The $cp\vartheta$ function value is $n$ if and only if $G_1$ and $G_2$
are isomorphic. Moreover, in this case the feasible region of CP-LT is equal to 
${\cal P}_{G_1G_2}$.
\end{lemma}

\begin{proof} Consider a non-negative vector realization $\{u_{ij}|i,j\in[n]\}\cup \{w\}$ 
for a CP solution $Y$ and the
corresponding matrix $W$ defined towards the end of Section 3. Since objective function 
attains value $n$, from Observation \ref{obs1} vectors of each row/column form a maximal set of pairwise
orthogonal united vectors. Also from the same Observation each row and each column adds up to $w$.
Assume that the vector realization is in an $N$-dimensional space.
 Consider the $r$-th component of the matrix, i.e., the matrix formed by
the $r$-th component of each vector. Let us denote it by $D_r$.
Each element of $D_r$ is non-negative and each row and each
column adds up to $w_r$, the $r$-th component of $w$. Hence $D_r$ is $w_r$ times a 
doubly-stochastic matrix. But the vectors
of the same row (resp. column) are orthogonal so exactly one entry is non-zero in each row
(resp. column) if $w_r >0$. So $D_r = w_rP_{\sigma_r}$ for some permutation $\sigma_r$.
We can express $W$ by $\sum_r w_rP_{\sigma_r}e_r$ where $e_r$ denotes the unit vector along the
$r$-th axis. $Y_{ijkl}$ is the inner product of the vectors $u_{ij}$ and $u_{kl}$
which is $(\sum_r w_r(P_{\sigma_r})_{ij}e_r)\cdot (\sum_s w_s(P_{\sigma_s})_{kl}e_s) 
= \sum_r w_r^2 (P_{\sigma_r})_{ij})(P_{\sigma_r})_{kl}) = 
\sum_r w_r^2(P^{[2]}_{\sigma_r})_{ij,kl}$. Thus $Y = \sum_r w_r^2 P^{[2]}_{\sigma_r}$. 
Since $\sum_rw_r^2 = w^2 =1$, $Y$ is a convex combination of some of the $P^{[2]}_{\sigma}$.
From Lemmas \ref{lem5} and \ref{lem3}
 each $\sigma_r$, with $w_r >0$, is an isomorphism between $G_1$ and $G_2$.
Since $w$ is a unit vector, $w_r>0$ for at least one $r$. Hence $G_1$ and $G_2$ are
isomorphic. Conversely from Lemma \ref{lem3} if $\sigma$ is an isomorphism, then $P^{[2]}_{\sigma}$
is a solution and its objective function value is $n$.

When the $cp\vartheta$ function has value $n$, the above discussion implies that
 the feasible region is contained in ${\cal P}_{G_1G_2}$.
Conversely, from Lemma \ref{lem3} and the fact that convex combination of CP solutions is also a CP solution, we deduce that ${\cal P}_{G_1G_2}$ is contained in the
feasible region. $\Box$
\end{proof}

\section{The Case of Non-Isomorphic Graphs}

\begin{lemma}\label{lem7} Let $Y$ be a solution of CP-LT. If $\sum_{j}Y_{ij,ij} \geq
1 - 1/(4n^4)$ for each $i$, then $G_1$ and $G_2$ are isomorphic.
\end{lemma}

\begin{proof}
Let $N$ denote the cp-rank of $Y$. So we have a non-negative vector realization of $Y$,
$\{u_{ij}| i,j \in [n]\}\cup \{w\}$ in an $N$-dimensional space. So we have
$u_{ij}\cdot u_{kl} = Y_{ij,kl}$ for all $i,j,k,l$ and there is an orthonormal basis
of this space, $B = \{e_p| p\in [N]\}$, such that every $u_{ij}$ belongs to
the closed positive orthant of this basis.

From Theorem \ref{thm1} $N < n^4$ because the rank of $Y$ is at most $n^2$.
Hence from the statement of this lemma $w\cdot \sum_ju_{ij} = \sum_ju_{ij}^2 = \sum_{ij}Y_{ij,ij}= 
 > 1 - 1/(4N)$ for each $i$.  
Let $S_i=\{u_{i1},\dots,u_{in}\}$ for each $i$. Since it is a set of
orthogonal united vectors with respect to $w$, $w.\sum_j u_{ij} =
\sum_j u_{ij}^2 \leq 1$. 

Without loss of generality assume that $w\cdot e_1\geq
w\cdot e_j$ for all $j$. So $w\cdot e_1 \geq 1/\sqrt{N}$ because $w$
is a unit vector.
If every vector in $S_i$ is perpendicular to $e_1$, then $w \cdot \sum_ju_{ij}$
can be at most $|w-(w\cdot e_1)e_1|$, which is at most
$(1-(w\cdot e_1)^2)^{1/2} \leq (1-1/N)^{1/2} \leq 1-1/(2N)$,
contrary to the given fact. So there exists a vector $u_{ij_i}\in S_i$
such that $u_{ij_i}\cdot e_1>0$. Let there be a $k\neq j_i$ such
that $u_{ik}\cdot e_1>0$. As all vectors are in the closed positive orthant,
$u_{ij_i}\cdot u_{ik}\geq (u_{ij_i}\cdot e_1)(u_{ik}\cdot e_1) > 0$.
This contradicts the fact that the vectors of $S_i$ are pairwise orthogonal.
Hence we conclude that for each $i$ there exists a unique vector
$u_{ij_i}\in S_i$ such that $u_{ij_i}\cdot e_1>0$. Thus $\sum_{j=1}^n
u_{ij}\cdot e_1 = u_{ij_i}\cdot e_1$.

Next we will show that $u_{ij_{i}}\cdot u_{kj_{k}} \geq 1/(16N^2)$ for
all $i,k\in [n]$. For any $i$, from the given facts $1-1/(4N) < w\cdot
\sum_ju_{ij} = (w\cdot e_1)(\sum_ju_{ij}\cdot e_1) +
(w-(w\cdot e_1)e_1)\cdot (\sum_ju_{ij}-(\sum_ju_{ij}\cdot e_1)e_1)$.
Since $(\sum_ju_{ij})^2 \leq 1$, $(w-(w\cdot e_1)e_1)\cdot (\sum_ju_{ij}-
(\sum_ju_{ij}\cdot e_1)e_1) \leq |(w-(w\cdot e_1)e_1)| \leq
1-1/(2N)$. The last inequality has been established in the previous paragraph.
So $(w\cdot e_1)(\sum_ju_{ij}\cdot e_1)\geq 1/(4N)$ for all $i$.
Hence $u_{ij_i}\cdot e_1 = \sum_ju_{ij}\cdot e_1  \geq 1/(4N)$.

All vectors of each $S_i$ are in the closed positive orthant hence
$u_{ij_{i}}\cdot u_{kj_{k}} \geq (u_{ij_{i}}\cdot e_1)(u_{kj_{k}}
\cdot e_1) \geq 1/(16N^2)$ for all $i,k\in [n]$.
Thus the set $\{u_{1j_1},\dots,u_{nj_n}\}$ is pairwise non-orthogonal, and hence
a complete consistent set. From the first paragraph of Section 4 we know that
 the permutation, $\sigma(i)=j_i$ for all $i$, is an isomorphism between
$G_1$ and $G_2$. $\Box$
\end{proof}

\begin{corollary} If $G_1$ and $G_2$ are non-isomorphic, then the value of the
$cp\vartheta$-function of CP-LT must be less than $n- 1/(4n^4)$.
\end{corollary}

\section{Conclusion}

We have seen that if $G_1$ and $G_2$ are isomorphic, then the $cp\vartheta$ function value
 is $n$. If the graphs are not isomorphic, then the function value
remains less than $n-1/(4n^4)$. Moreover in isomorphic case the feasible region is 
exactly ${\cal P}_{G_1G_2}$. 
Since a completely positive program cannot be solved in polynomial time, the above is not 
a polynomial time test. We have the following theorem.

\begin{theorem} The proposed completely positive $\vartheta$ function takes value $n$ when the graphs are 
isomorphic
and it takes value less than $n-1/(4n^4)$ when the graphs are non-isomorphic. This gives a test for graph isomorphism.
\end{theorem}

\bibliographystyle{plain}
\bibliography{reference}

\end{document}